\documentclass{llncs}

\usepackage{makeidx}  
\usepackage[title]{appendix}
\usepackage{graphicx}
\usepackage{amsmath}
\usepackage{amssymb}
\usepackage{enumitem}
\usepackage{mathtools}
\DeclarePairedDelimiter\ceil{\lceil}{\rceil}
\DeclarePairedDelimiter\floor{\lfloor}{\rfloor}

\begin{document}

\setlength{\abovedisplayskip}{-6pt}
\setlength{\belowdisplayskip}{2pt}

\mainmatter              
\title{Nakamoto Consensus with Verifiable Delay Puzzle}
\titlerunning{}  
%
\author{Jieyi Long}
\authorrunning{} 
%
\tocauthor{}
\institute{Theta Labs, San Jose CA 95128, USA\\
\email{jieyi@thetalabs.org}}

\maketitle              
\let\labelitemi\labelitemii

\begin{abstract}
This paper presents a new consensus protocol based on verifiable delay function. First, we introduce the concept of verifiable delay puzzle (VDP), which resembles the hashing puzzle used in the PoW mechanism but can only be solved sequentially. We then present a VDP implementation based on the continuous verifiable delay function. Further, we show that VDP can be combined with the Nakamoto consensus in a proof-of-stake/proof-of-delay hybrid protocol. We analyze the persistence and liveness of the protocol, and show that compared to PoW, our proposal consumes much less energy; compared to BFT leader-election based consensus algorithms, our proposal achieves better resistance to long-range attacks and DoS attacks targeting the block proposers.
\end{abstract}
\section{Introduction}

Since the debut of Bitcoin \cite{nakamoto2008bitcoin} in 2009, the energy consumption of its \textit{proof-of-work} (PoW) consensus protocol has been growing at a stunning pace. It is estimated that the electricity demand of Bitcoin has already surpassed that of a small country such as Denmark. This has motivated researchers to investigate alternative consensus protocols that are more energy efficient and environmentally friendly. Among different replacement candidates, the \textit{proof-of-stake} (PoS) mechanism has attracted widespread attention \cite{bentov2014wopow,daian2016snowwhite,buchman2018tendermint,buterin2016posfaq,chen2018algorand,gilad2017algorand,micali2016algorand,buterin2017casperffg,kiayias2017ouroboros}. Instead of spending intensive computational resources on solving hashing puzzles, PoS protocols typically run a computationally inexpensive process to randomly select the next block proposer based on the stake of each node.

There are two major camps in PoS mechanism design. The first is the chain-based PoS which simulates the PoW process for leader election. While this approach is energy-efficient, some of the early works were insecure. In particular, they were vulnerable to the ``nothing-at-stake" attack since the cost for simulating PoW is minimal \cite{buterin2017casperffg}. Later protocols including Ouroboros address such concerns \cite{kiayias2017ouroboros}. However, the leader election in Ouroboros employs a secure multi-party computation based coin-tossing scheme that incurs a relatively large overhead. This scalability bottleneck is addressed in Ouroboros Praos \cite{david2017ouroborospraos} via a more efficient leader election scheme. A recent improvement, Ouroboros Genesis \cite{badertscher2018ouroborosgenesis} further allows offline parties to safely bootstrap the blockchain when they come back online. The second type of design stems from the traditional Byzantine Fault Tolerant (BFT) research \cite{lamport1982bgp,fischer1985flp,dwork1988consensuspartialsync,castro1999pbft}. Algorithms such as Tendermint, Casper FFG, Algorand, and HotStuff fall into this category \cite{buchman2018tendermint,buterin2017casperffg,chen2018algorand,gilad2017algorand,micali2016algorand,abraham2018hotstuff}. The BFT based approaches usually have proven mathematical properties. For example, under the assumption that more than 2/3 stake are held by honest participants, the safety property can typically be guaranteed even with an asynchronous network \cite{castro1999pbft}. Moreover, in the BFT based algorithms, the aforementioned attacks can be mitigated with the ``slashing condition" introduced by the Casper FFG \cite{buterin2017casperffg}. However, due to communication complexity, traditional BFT based algorithms are less scalable. Typically they only allow a relatively small number (usually less than 100) of nodes to participate in the consensus process, which limits the level of decentralization, making them less attractive for permissionless public chains. Algorand \cite{chen2018algorand,gilad2017algorand,micali2016algorand} improves the scalability with the introduction of cryptographic sortition to select a relatively small committee for each round, and runs the BFT consensus only within the committee. HotStuff \cite{abraham2018hotstuff}, on the other hand, reduces the communication complexity by signature aggregation and pipelining the voting phases. 

Although these state-of-the-art PoS protocols effectively reduces the energy consumption, they are potentially more vulnerable to the so-called ``long-range attack'' compared to the PoW alternative. In a long-range attack, the adversary forks the blockchain starting from a past block, and adds blocks to his private branch in order to overtake the longest public chain \cite{deirmentzoglou2019longrangesurvey}. PoS-based blockchains are generally more susceptible to long-range attacks, since unlike PoW, minting a block is almost costless in PoS. A particular threat to the PoS blockchain is ``posterior corruption", a special form of long-range attack \cite{deirmentzoglou2019longrangesurvey}. If an adversary can accumulate the majority stake at a certain past block height through posterior corruption, from that block he can fabricate an alternative fork indistinguishable from the longest public chain in a short amount of time, since generating a block requires minimal computational effort. Ouroboros Praos and Genesis \cite{david2017ouroborospraos,badertscher2018ouroborosgenesis} proposed to use key evolving signature schemes \cite{franklin2006keyevolving} as a countermeasure, but it assumes that without incentive, the nodes would voluntarily delete their ephemeral private keys after each use, which is not necessary practical.

A question thus arises: Can we combine the merits of PoW and PoS? To be more specific, \textit{our goal is to design a consensus protocol which has good resistance to the long-range attacks, and yet consumes low amount of energy}. In the paper we will present a new \textbf{proof-of-stake/proof-of-delay hybrid} protocol based on the concept of \textit{verifiable delay puzzle} (VDP), which can hopefully achieve these goals. In this framework, VDP resembles the hashing puzzle used in the PoW schemes but can only be solved sequentially. Our VDP construction is based on recent advancements of the \textit{verifiable delay functions} (VDF) research \cite{boneh2018vdf,pietrzak2018svdf,wesolowski2018evdf,boneh2018survey2vdfs,ephraim2020cvdf}. VDF is a type of function that requires a specified number of \textit{sequential} steps to evaluate, and produces a unique output that can be publicly and efficiently verified \cite{boneh2018vdf}. At the first glance, such difficult-to-evaluate but easy-to-verify asymmetry is very similar to the PoW hashing puzzle, making VDF a perfect drop-in replacement for the PoW hashing puzzle. However, different from the hashing puzzle, a VDF function in its original form has no intrinsic randomness. Given the same input, a VDF function always outputs the same result regardless of who computes it. To fully emulate the hashing puzzle, unpredictability needs to be somehow injected into the VDF to create the lottery effect to determine the next block proposer. Furthermore, to achieve the energy saving goal, the consensus protocol also needs to be designed in a way such that unlike in PoW, equipping with more VDF solvers does not give a node proportional advantage. Later in the paper, we will show how these properties can be achieved altogether. Incorporating VDP also makes the consensus protocol more resilient to the long-range attacks than the existing PoS designs, since proposing a block requires solving a VDP, which takes non-negligible amount of time even on a parallel computer.
\bigbreak
\noindent \textbf{Our contributions.} Aside from introducing the concept of VDP, we also make the following novel contributions:

\begin{itemize}
  \item We present a VDP construction based on the continous VDF design \cite{ephraim2020cvdf}, and incorporate the VDP into the Nakamoto consensus in a PoS blockchain. To create the lottery effect, we propose to leverage the \textit{verifiable random function} (VRF) to pseudorandomly assign different verifiable delay puzzles to the nodes competing to propose the next block. The parameters of the puzzle for a new block are fully determined by the private key of the node and the parent block. As a result, each node could take a different amount of time to solve its assigned puzzle instance, and thus randomness arises. Such a combination of VDF and VRF also makes it very efficient for other nodes to verify that a node has indeed solved its assigned puzzle instance.
  \item We propose a set of slashing rules and show that with these rules, an economically rational validator node is worse-off if it publishes blocks on multiple forks. Thus, the ``nothing-at-stake" attack can be mitigated by our protocol.
  \item We analyze the protocol backbone and derive three basic properties, namely, chain growth, chain quality, and common prefix. From these properties we prove that our protocol can achieve persistence and liveness. 
  \item The last but not the least, we analyze the protocol's energy consumption and its resistance to various forms of long-range attacks, and thus show that the proposed protocol achieves our design goals. 
  
\end{itemize}

\bigbreak
\noindent \textbf{Related works.} We note that there are other attempts to incorporate VDF in consensus protocols, including using VDF for creating random beacons, and for validator shuffling in blockchain sharding. Among these works, a protocol called Chia proposed by Cohen et al., combines proofs-of-space with VDF \cite{cohen2016post,cohen2019chia}. In their proposal, a space miner (called ``farmer'' in \cite{cohen2019chia}) can generate a proof which demonstrates that it has access to a certain amount of disk space. The miner then maps the space proof to a random integer $\tau$ in a way that a larger disk space has a higher likelihood to generate a smaller $\tau$. The miner then evaluates VDF for $\tau$ iterations. The first miner that successfully computes the VDF can broadcast its block. Similar to our proposal, this process mimics the random delay weighted by the disk space possessed. However, since $\tau$ is explicitly calculated, a miner might be reluctant to compute the VDF if its $\tau$ is too large or could perform some kinds of grinding attack with this information. This could create bias and weaken the security of the system. In contrast, in our VDP based approach, a node cannot estimate the time to solve the assigned delay puzzle until it actually solves it, which is more akin to the hashing puzzle in PoW. Furthermore, our approach is based on proof-of-stake instead of proof-of-space, which does not require the time-consuming storage setup phase. Proof-of-stake also has less moving parts and thus is simpler to analyze and implement. Finally, our analysis shows that our protocol only requires \textit{economically rational} nodes to control the super majority stake, instead of \textit{honest} nodes as required by Chia. This makes our protocol more suitable for real world applications such as cryptocurrencies. 

Ethereum 2.0 plans to employ RANDAO, a VDF based public randomness beacon for block proposer election \cite{drake2018randaomvp,drake2018vdflookahead}. While the RANDAO design can generate unbiased randomness, the public nature of the beacon indicates that everyone in the network can predict the next or even the next few block proposers. The adversary could exploit this information to launch targeted attacks. In contrast, our protocol is driven by private randomness similar to Algorand \cite{chen2018algorand}. Hence the adversary does not know which node will become the block proposer until the next block gets published. This makes our protocol more resilient to DoS attacks.
\section{Verifiable Delay Puzzle}

In this section we introduce the concept of verifiable delay puzzle, and present a construction based on the continuous verifiable delay function \cite{ephraim2020cvdf}. Informally, a Verifiable Delay Puzzle is a puzzle which requires at least $t$ sequential steps to solve even on a parallel computer, while its solution can be verified in $O(\textrm{poly}(\log(t)))$ steps. It is analogous to the hash puzzle in the proof-of-work scheme, but can only be solved sequentially.

In the following, as in Boneh et al. \cite{boneh2018vdf}, we say that an algorithm runs in parallel time $t$ with $p$ processors if it can terminate in time $t$ on a PRAM machine with $p$ parallel processors. We use the term total time (eq. sequential time) to refer to the time needed for computation on a single processor. We define VDP as follows:

\begin{definition} \label{def:vdp}
Given a security parameter $\xi$, a Verifiable Delay Puzzle (VDP) implements a function $f: X \rightarrow Y$ which maps a non-negative integer to another integer. Furthermore, given an integer $m \in (\min(Y), \max(Y))$, the VDP asks for a solution which is a pair $(t, \pi)$ with the following properties:

\begin{itemize}
  \item \textbf{Efficient verifiability}: In the solution, $t$ is a non-negative integer, and $\pi$ is a proof. There exists a deterministic algorithm $\mathcal{V}(f, t, \pi) \rightarrow \{Yes, No\}$, which verifies $f(t) < m$ with the proof $\pi$. This verification algorithm must run in total time polynomial in $\log(t)$ and $\lambda$, i.e., in $O(\textrm{poly}(\log(t), \lambda)$ steps.
  
  \item \textbf{Sequentiality}: Any parallel algorithm $\mathcal{A}$ that can produce a solution $(t, \pi)$ which satisfies $\mathcal{V}(f, t, \pi) = Yes$ would take at least $t$ sequential steps to generate the solution, using at most $\textrm{poly}(\xi)$ processors.
\end{itemize}
\end{definition}


One interesting aspect of this definition is that the number of steps $t$ needed to solve the puzzle is \textit{unknown} before the puzzle is solved. This property is crucial for preventing grinding attacks. Moreover, it is worth pointing out that with the \textit{sequentiality} requirement, an adversary with a large number of parallel processors has no advantage compared to a node with a single processor if the processors have the same speed.

\bigbreak
\noindent \textbf{VDP construction based on continuous VDF.} Pietrzak and Wesolowski separately proposed two simple VDF constructions based on repeated squaring \cite{pietrzak2018svdf,wesolowski2018evdf}. However, both constructions require that the total number of VDF evaluation cycles is known when initializing the VDF. Very recently, Ephraim et al. presented a continuous VDF design which is also based on repeated squaring, but allows efficient proofs of any intermediate $t$ during the VDF evaluation \cite{ephraim2020cvdf}. A VDP can thus be constructed based on the continuous VDF: \textit{Given input $r$ and a constant $\gamma$, find a solution $(t, d, \pi^{d})$, such that}

\begin{equation} \label{eq:repeated_squaring_vdf}
  d = H(r)^{2^t}
\end{equation}
\begin{equation} \label{eq:hash_threshold}
  K(d) \leq \gamma \cdot M
\end{equation}

Here $H(\cdot)$ is a hash function acting as a random oracle. And the constant $\gamma$ is a threshold which controls the difficulty of the puzzle. $K(\cdot)$ is a one-way hash function which maps its input to a non-negative integer, and $M$ is the maximum value of the hash function.

\begin{theorem} \label{th:vdp-construction}
The puzzle defined above is a verifiable delay puzzle.
\end{theorem}

\begin{proof} 
To see why this is a verifiable delay puzzle, we define the following notations: function $f(t) = K(H(r)^{2^t})$, $m = \gamma \cdot M$, and $\pi = (d, \pi^{d})$. 

First, we note that the above construction implements a function which maps an integer $t$ to another integer $K(H(r)^{2^t})$ as required by the definition. We claim that function $f(t)$ and pair $(t, \pi)$ satisfy the two properties in Definition \ref{def:vdp}.

To prove the \textit{efficient verifiability} property, we note that to check if $(t, \pi) = (t, (d, \pi^{d}))$ is indeed a solution to the VDP, any third party just needs to verify 1) $\pi^{d}$ proves $d = H(r)^{2^t}$, and 2) $K(d) \leq m = \gamma \cdot M$. Continuous VDF allows verification of condition 1) for any $t$ value in $O(\textrm{poly}(\log(t))$ \cite{ephraim2020cvdf}, and verification of condition 2) takes $O(1)$ time. Thus, the entire verification can obviously be done in $O(\textrm{poly}(\log(t))$ time.

Next, to prove the \textit{sequentiality} property, we note that since $d$ is required to be a part of the solution, solving this VDP involves explicitly calculating the $d$ value and generating $\pi^{d}$ which proves that $d = H(r)^{2^t}$. If $t$ is the smallest integer that satisfies (\ref{eq:repeated_squaring_vdf}) and (\ref{eq:hash_threshold}), we claim that solving the puzzle takes at least $ts$ sequential steps, even for an adversary with $\textrm{poly}(\xi)$ parallel processors. Otherwise, if an adversary is able to come up with the VDP solution $(t, (d, \pi^{d}))$ in less than $t$ steps, then he already solved the repeated squaring VDF $d = H(r)^{2^t}$ in less than $t$ steps, which contradicts with the conclusions of  \cite{pietrzak2018svdf,ephraim2020cvdf}. Therefore, this construction conforms to our definition of a VDP.
\end{proof}

It is worth pointing out that the sequentiality property rules out the possibility of using proof-of-work style brute-force methods to ``guess" the solution of a VDP. In particular, with parallel computing resources, an adversarial party may be able to test multiple values of $d$ simultaneously to see if $K(d) \leq \gamma \cdot M$. However, for a given $d$, to find the value of $t$ such that $d = H(r)^{2^t}$ still takes at least $t$ sequential steps (if such a $t$ exists).

\section{System Model} \label{sec:system-model}

Before presenting the protocol design, we first provide the system model and our assumptions. Note that different VDP solvers might have different speeds. A faster VDP solver could give a node advantages over others. Considering the speed variation, we introduce the concept of speed-weighted stake below.

\begin{definition} \label{def:speed-weighted-stake}
Suppose a validator node $v_{i}$ owns $s_{i}$ fraction of the total stake, and its VDP solver can conduct $q_{i}$ sequential VDF evaluation steps per unit time. Assuming there are $n$ validators in total, the \textbf{speed-weighted stake} $sws_{i}$ of validator node $v_{i}$ can be defined as

\begin{equation} \label{eq:speed_weighted_stake}
  sws_{i} = \frac{v_i \cdot q_i}{\sum^{n}_{j}{v_j \cdot q_j}}
\end{equation}

\end{definition}


\textbf{Validator model}: The above definition refered a type of blockchain node called the \textit{validator}. Validators are the block producers of the blockchain network. In order for a node to become a validator, a certain amount of tokens needs to be staked to the node. We will describe the validator role in more details later, but here we assume that strictly more than $1 - \frac{1}{1+e} \approx 73.1\%$ of the total speed-weighted stakes are controlled by \textit{economically rational} validator nodes, where $e$ is the Euler's constant. In other words, $\sum_{i \in RV}{sws_i} > 1 - \frac{1}{1+e}$, where $RV$ is set of all rational validators. The $1 - \frac{1}{1+e}$ bound is derived from Theorem \ref{th:honest-adversarial-gap-bound} which will be presented in Section \ref{sec:protocol-analysis}. Furthermore, Theorem \ref{th:nothing-at-stake-immune} indicates that a rational validator should execute the protocol as prescribed. Thus we will use the term \textit{honest rational validators} interchangeably with economically rational validators, or simply \textit{honest validators}.
  
\textbf{Attacker model}: We assume powerful byzantine attackers. They can corrupt many targeted nodes, and have access to a large number of VDP solvers (but bounded by $\textrm{poly}(\xi)$). However, we assume that they can control no more than $\frac{1}{1+e}$ of the total speed-weighted stakes. For example, the attackers might have VDP solvers that are 3 times as fast as the average of the rational nodes. Yet if they controls less than 10\% of the total stake, their speed-weighted stake fraction is 0.25 ($= 10\% \cdot 3 / (10\% \cdot 3 + 90\% \cdot 1)$), which is still less than $\frac{1}{1+e}$. Also, we assume the attackers are computationally bounded. For instance, they cannot forge fake signatures, and cannot invert cryptographic hashes.
  
\textbf{Timing and communication model}: To focus on the core problems, we will present and analyze our proposed protocol in an ideal setting where the participants operate in a synchronous communication network with zero latency. We do not assume a direct message channel between all pairs of validators. Messages between them might need to be routed through other nodes, some of which could be byzantine nodes. 

%
\section{Nakamoto Consensus with VDP}

\subsection{Basic Design}

For simplicity, let us first consider a blockchain with a fixed set of validators who are eligible to propose new blocks. Now, let us see how to can leverage VDP to create the ``lottery effect" to pseudorandomly select a validator to propose the next block. We note that in Formula (\ref{eq:repeated_squaring_vdf}) the VDP takes an input $r$. With different values of $r$, the amount of time needed to solve the VDP varies. This is where we can inject the randomnes. In particular, we can leverage the \textit{verifiable random function} (VRF) \cite{gilad2017algorand,micali2016algorand} to pseudorandomly determine the $r$ for each validator. As formulated in Formula (\ref{eq:vrf}), a VRF typically takes a private key $sk$ of a node and a $seed$ as input, and uniquely generates a pseudorandom random value $r$ and the corresponding proof $\pi^{r}$. The proof $\pi^{r}$ enables anyone that knows the node's public key $pk$ to check that $r$ indeed corresponds to the $seed$, without having to know the private key $sk$.

\begin{equation} \label{eq:vrf}
  (r, \pi^{r}) \leftarrow VRF_{sk}(seed)
\end{equation}

Assume the block at the tail end of the chain has height $i-1$ as shown in Fig. \ref{fig:validator_compete_to_solve_VDPs}. Validator $u$ wants to propose a block for height $i$. To do this it first needs to solve the following VDP instance:

\begin{equation} \label{eq:vdp_basic_1}
  (r_{i}, \pi_{i}^{r}) \leftarrow VRF_{sk_u}(seed = d_{i-1})
\end{equation}
\begin{equation} \label{eq:vdp_basic_2}
  (t_{i}, d_{i}, \pi_{i}^{d}) \leftarrow VDP(r = r_i)
\end{equation}

\begin{figure}[htb]
\begin{center}
\includegraphics[width=0.85\textwidth]{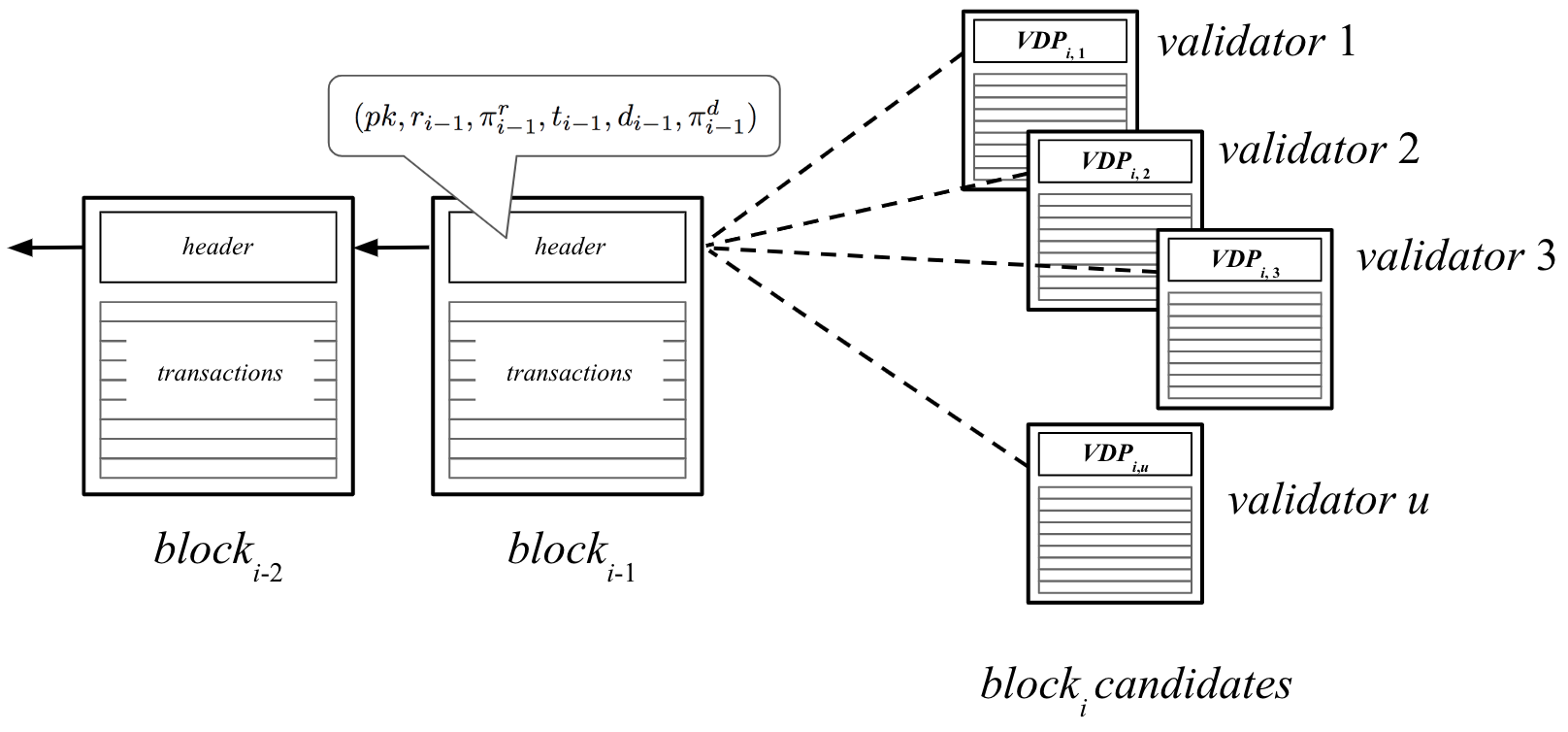}
\end{center}
\caption{Validators compete to solve VDPs. Each validator node is assigned a VDP instance fully determined by the $d$ value recorded in the parent block and the private key of the node. The VDP solution tuple is embedded in the header of the blocks.}
\label{fig:validator_compete_to_solve_VDPs}
\end{figure}

In Formula (\ref{eq:vdp_basic_1}), $sk_u$ is the private key of validator $u$. We note that the VRF takes $d_{i-1}$ as the seed. It is the $d$ value in the VDP solution for the parent block at height $i-1$. The VRF generates $r_i$, which is then used as the input for the VDP in Formula (\ref{eq:vdp_basic_2}). We note that since each validator node has its own private key $sk$, the pseudorandom value $r_i$ generated by Formula (\ref{eq:vdp_basic_1}) is unique for each validator. Thus, each validator needs to solve a different VDP instance specified by Formula (\ref{eq:vdp_basic_2}) which could require a different number of VDF evaluation cycles. This is where the randomness arises. 

Similar to Bitcoin, in our protocol, to propose the next block, a validator needs to first solve its VDP. Fig. \ref{fig:validator_compete_to_solve_VDPs} illustrates the process where the validators compete to solve VDPs. Once the assigned VDP is solved, a validator can broadcast a block. The header of the new block should contain the following tuple

\begin{equation} \label{eq:block_header_tuple}
  (pk, r_{i}, \pi_{i}^{r}, t_{i}, d_{i}, \pi_{i}^{d})
\end{equation}

Besides, the block header should also contain $\sigma_{sk}(BH_{i})$, the validator's signature of the block header (which includes the above tuple, the hash of the parent block, and the Merkle root of the transactions trie).  With the signature, if a validator tries to fork the chain by proposing multiple blocks with the same VDP solution but different sets of transactions, it can be detected by others and get punished (see the \textbf{slashing rules} in Section \ref{sec:slashing-rules}). Furthermore, this also prevents a malicious validator from proposing a block using another validator's VDP solution. If a malicious validator creates a block $B_{i}^{\prime}$ which steals the VDP solution from block $B_{i}$ but with a different set of transactions, the signer of the signature $\sigma_{sk}(B_{i}^{\prime})$ would not match with the prover of the VDF proof $\pi_{i}^{d}$, thereby rendering $B_{i}^{\prime}$ an invalid block. 

As soon as a node receives a new block, it should first verify $(r_{i}, \pi_{i}^{r})$ against the block proposer's public key $pk$ and seed $d_{i-1}$. This is to confirm whether the block proposer was indeed solving the VDP instance assigned by the VRF. If the check is passed, the node then verifies the VDF puzzle solution $(t_{i}, d_{i}, \pi_{i}^{d})$. Also the node verifies the signature $\sigma_{sk}(B_{i})$ to make sure the block is proposed by the validator with public key $pk$. If all checks are passed, and all the transactions contained in the block are valid, the validator can append this block to the block tree in its \textit{local view}.

It is worth pointing out that in the above description, the seed $d_{i-1}$ does not depend on the content of the block (i.e. the transactions included in the block). This means for a validator, if it has decided to ``mine" on top of a block, the VDP instance it needs to solve is \textbf{fully determined}. Hence, a validator cannot \textit{grind} by changing the transactions in his proposed block to gain extra advantages. 

Although the grinding attack can be addressed by making the seed in the VDP independent of the content (i.e., the transactions) of the block, there is a potential loophole an attacker can exploit. Since there is no binding between the VDP proof and the transactions included a block, an attacker might attempt to replace the content of the block after the fact. If this can be done, different chains with the same VDP proofs but different block contents could emerge. Fortunately, as mentioned above, the block needs to contain the VDP solution, the validator's signature of the block header, and the validator's public key. When a validator verifies a block, it should check the VDP solution and the block signature with the validator's public key. If a malicious validator wants to replace the blocks after the fact, since each block contains the hash of the parent block, he has to forge all the blocks starting from the divergence point. To forge these blocks, indeed he can ``reuse" the VDP results which are independent of the block content. However, as mentioned above, the input to the VDP (i.e. $r_i$) needs to be verified against the block proposer's public key. Thus, the VDP of a block is tied to the public key of the proposer. And that same public key should be able to verify the signature included in the same block, which signs the Merkle root of the transaction trie. In other words, \textit{the VDP is indirectly tied to the block content} through the public key of the block proposer. According to the chain quality lemma (Lemma \ref{lm:chain-quality} in Section \ref{sec:protocol-analysis}), it is guaranteed that a portion of the longest public chain are honest blocks. Without the private keys of these honest block proposers, the attacker cannot generate valid block signatures if he modifies the block content.

The \textbf{incentive structure} is similar to that of Bitcoin. Each block comes with block reward for the block proposer in the form of newly minted tokens. Also each transaction may carry a certain amount of fees that the block proposer can collect. To \textbf{resolve forks}, we adopt the Nakamoto longest chain rule. Bitcoin requires a number of block confirmations (typically 6 confirmations) to ensure a transaction is safe. Similarly, we have the following definition:


\begin{definition} \label{ref:confirmed-block}
A validator considers a block as \textbf{confirmed} if on its local view, the block is on a chain where there are least $\kappa_{con}$ valid blocks on the chain after it, where $\kappa_{con}$ is a security parameter.
\end{definition}

Compared to some of the BFT-based PoS protocols that \textit{publicly} elect one or a small set of leaders each round to propose blocks \cite{buchman2018tendermint,abraham2018hotstuff}, our proposed protocol also has the unique advantage that it is much more immune to targeted DoS attacks to the block proposers. This is because, similar to Bitcoin, anyone solves its assigned VDP instance is eligible to propose a new block. A DoS attacker might be able to shut down a portion of the network, but the remainder nodes can still execute the protocol to extend the blockchain.

\subsection{Permissionless Blockchain} \label{sec:permissionless-blockchain}

Now we can extend the discussion to a permissionless proof-of-stake blockchain where any node can become a validator after putting down some stake. To be more specific, to become a validator, first a node needs to generate a private/public key pair $(sk, pk)$, and then stake a certain number of tokens to its public key. To prevent sybil attack and also to avoid power concentration, we require each validator stakes \textbf{a fixed amount of tokens} $S$. For example, $S$ can be set to $1/10000$ of the total token supply so the system can accommodate up to $10,000$ validators. We also require the staked tokens to be \textbf{locked} for a certain duration (e.g. at least four weeks) before they can be unlocked.

With this protocol enhancement, any user with $S$ amount of stake can run a validator node. If a user possesses more than $S$ amount of stake, he can run multiple validators with his stake divided. We note that compared to the BFT voting based consensus mechanisms \cite{buchman2018tendermint,abraham2018hotstuff}, the chain-based consensus mechanism eliminates explicit voting and the associated communication cost. Thus, our approach is as scalable as Bitcoin in terms of the number of validators that can participate in the consensus protocol, and thus can achieve a very higher level of decentralization.

\subsection{Slashing Rules} \label{sec:slashing-rules}

To deter economically rational validators from forking the chain, we introduce the following rules inspired by Ethereum \cite{buterin2016posfaq} and SpaceMint \cite{park2018spacemint}:

\begin{itemize}
\item If a validator publishes two blocks with the same parent block, its entire stake deposit will be slashed.
\item If a validator publishes blocks on multiple forks, then on a particular fork, its deposit will be deducted by $(1 + \epsilon) \cdot R$ for each block on other forks, where $R$ is the block reward, and $0 < \epsilon \ll 1$. The deducted $(1 + \epsilon) \cdot R$ tokens are burnt completely.
\item As an incentive, the submitter of the slashing transaction gets $\epsilon \cdot R$ newly minted tokens as submitter reward for each block.
\end{itemize}

Here the term \textit{slashing transaction} refers to a special type of transaction \footnote{In an Ethereum-like smart contract platform, staking and slashing can also be implemented using a smart contract, where staking sends the required amount of ETH to the smart contract as collateral. A slashing transaction is a call to the smart contract with the necessary proofs. Once the proofs are verified by the smart contract, a proper amount of staked ETH can be seized.} that can carry the proof that a certain validator has signed and published blocks on conflicting forks. It also contains the public key and signature of the submitter for it to claim the $\epsilon \cdot R$ submitter reward. Note that for the same block, only the first reporter gets the submitter reward.

We note that most proof-of-stake protocols do not slash a validator for publishing blocks on multiple forks. Thus, at the first glance, our proposed slashing rules seem to pose more risks for the validators. However, it is worth pointing out that the proof-of-work mechanism also implicitly penalizes the miners that extend multiple forks. This is because only one single fork will win in the end, and hence all the electricity spent mining on the other forks is wasted.

Moreover, if we tune the VDP parameters properly (setting the value of $\gamma$ in Formula (\ref{eq:hash_threshold}) to a small value), the probability that two or more validators solve their assigned puzzle within one block propagation time can be made very small, similar to Bitcoin. Under such conditions, we can prove the following theorem, which states that a rational node should only propose new blocks on one chain, i.e. the longest public chain. The proof will be provided in Appendix \ref{sec:analysis-of-slashing-rules}.

\begin{theorem} \label{th:nothing-at-stake-immune}
For a rational node, proposing new block only on a fork that is most likely to win (i.e. the longest public chain in its local view) is a dominant strategy. Thus, the proposed protocol with the slashing rules is immune to the ``nothing-at-stake'' attacks.
\end{theorem}

\section{Protocol Analysis} \label{sec:protocol-analysis}

The analysis in this section will be conducted under the system model and assumptions provided in Section \ref{sec:system-model}. Most of the analysis revolves around the concept of \textbf{speed-weighted stake} introduced earlier. Regarding speed-weighted stake, we have the following lemma (proof provided in Appendix \ref{sec:proofs-backbone}).

\begin{lemma} \label{lm:block-rate-speed-weight-stake}
Denote the total speed-weighted stake of a party by $W = \sum_{i \in V}{sws_i}$, where $sws_i$ is the speed-weighted stake of validator $v_i$ as defined in Definition \ref{def:speed-weighted-stake}, and $V$ is the set of validators a party controls. Then, following the honest behavior, the aggregated block production rate $\lambda$ of these validators is proportional to $W$.
\end{lemma}

\subsection{The Protocol Backbone Analysis}

Taking inspiration from Garay et al.'s analysis of the Bitcoin protocol \cite{garay2015bitcoinbackbone}, we introduce the \textit{backbone} of our proposed protocol, and analyze its three important properties: \textbf{chain growth}, \textbf{chain quality}, and \textbf{common prefix}. From these properties we can further derive the persistence and liveness of our protocol using Garay et al.'s framework.

Before proving these three properties, we would like to remark on one important difference between the PoW puzzle and VDP. It is well-known that for a PoW miner, splitting the hashing power on multiple forks is inferior to concentrating all the hashing power on a single fork. However there is a subtle approach to solve multiple VDPs in parallel to speed up the chain growth. If a party control multiple validators, he can potentially perform a \textbf{breadth-first-search} to find a fork that can outgrow the longest public chain. This approach is illustrated in Fig. \ref{fig:bfs_attack}. In this example, the adversary controls two validators, and has a large number of VDP solvers available. At height $i$, the adversary can create two forks, each by one valiator. Then at height $i+1$, he creates four forks, two from each fork at height $i$. The adversary can \textit{secretly} expand such forks into a tree structure without being punished by the slashing rules. He then publishes the fastest growing fork if it is longer than the longest public chain. 
\begin{figure}[htb]
\begin{center}
\includegraphics[width=0.8\textwidth]{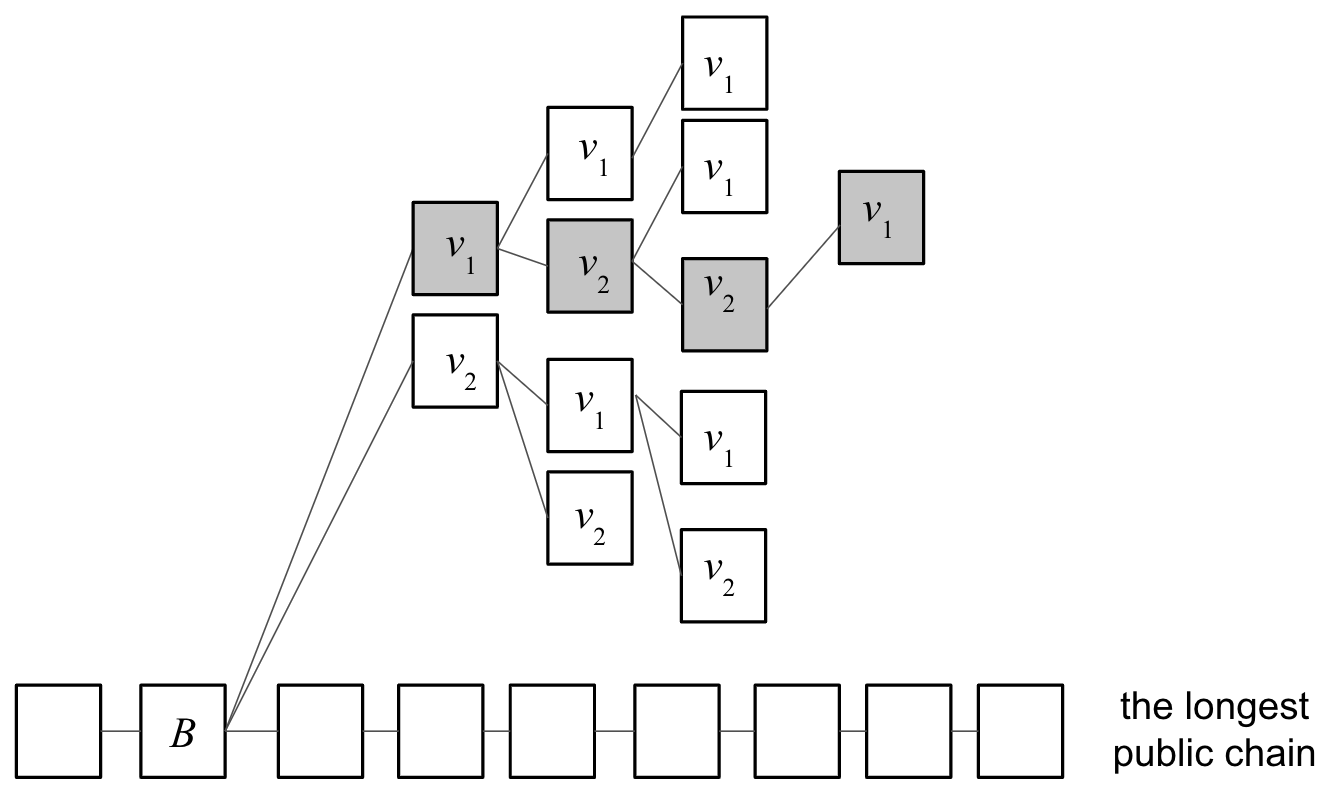}
\end{center}
\caption{Breadth-first-search to find the fastest growing fork. Assume the adversary controls two validators $v_1$ and $v_2$, and has an unlimited number of VDP solvers. He can expand the ``search tree" rooted at block $B$ and identify the fastest growing branch (marked in gray). All the VDPs at the same tree depth can be solved in parallel. For illustrative purpose, this figure only shows a case with branching factor 2. However, the analysis in this section applies to any arbitrarily large branching factor (e.g. 1,000).}
\label{fig:bfs_attack}
\end{figure}

Fortunately, even with this strategy, the growth rate of the longest branch (in terms of block height) of the tree is bounded compared to the \textbf{honest behavior} as prescribed by the protocol. We can apply the branching random walks analysis as in Bagaria et al. \cite{bagaria2019poslc} to bound the growth rate of the longest adversarial branch with the following lemma:

\begin{lemma} \label{lm:bfs-speed-up-bound} 
Assume that the adversary builds a block tree rooted at a block $B$ on the longest public chain for $T$ amount of time. Let $D_a(T)$ represents the number of blocks on the longest branch of the adversarial block tree, and $\lambda_a$ represents the block production rate \textbf{if the adversary follows the honest behavior}, then the following inequality holds true, where $e \approx 2.718$ is the Euler's constant:

\begin{equation} \label{eq:D_a_t_bound}
Pr(D_a(T) > e \lambda_a T + x) \leq e^{-x}
\end{equation}

\end{lemma}
 
This lemma is equivalent to Lemma 4 in Bagaria et al. \cite{bagaria2019poslc}, and the proof is omitted here. It indicates that with the same amount of speed-weighted-stake, the probability that the breadth-first-search approach can identify a chain with the \textit{expected growth rate} at most $e$ times as high as that of the honest strategy, regardless of the branching factor. In fact, following the arguments in \cite{bagaria2019poslc}, we can prove that \textit{no matter what block proposal strategy the adversarial parties might employ}, the bound in Lemma \ref{lm:bfs-speed-up-bound} holds true. Next, we provide the following two lemmas to characterize the longest chain growth. The proofs for both lemmas are provided in Appendix \ref{sec:proofs-backbone}.

\begin{lemma} \label{lm:number-of-honest-blocks-bound}
Denote the block production rate of the honest parties by $\lambda_h$. Assume that within $T$ amount of time, $N_h(T)$ honest blocks are proposed, then for any $\delta \in (0, 1)$, we have:

\begin{equation} \label{eq:D_h_t_bound}
Pr(N_h(T) < (1 - \delta) \lambda_h T) \leq 2 e^{- \lambda_h T \delta^2 / 3}    
\end{equation}

\end{lemma}

\begin{lemma} \label{lm:longest-public-chain-growth-rate-bound}
Denote the block production rate of the honest parties by $\lambda_h$. Assume that after $T$ amount of time, the longest chain grows by $D_h(T)$ blocks, then for any $\delta \in (0, 1)$, we have:

\begin{equation} \label{eq:D_h_t_bound}
Pr(D_h(T) < (1 - \delta) \lambda_h T) \leq 2 e^{- \lambda_h T \delta^2 / 3}    
\end{equation}

\end{lemma}

The theorem below indicates that with high probability, the gap between the longest public chain and the longest adversarial private chain increases linearly with the block height. The proof for the theorem is provided in Appendix \ref{sec:proofs-backbone}.

\begin{theorem} \label{th:honest-adversarial-gap-bound}
Assume starting from the genesis block, the honest rational parties always control strictly more than $1 - \frac{1}{1+e} \approx 73.1\%$ fraction of the total speed-weighted-stake. Consider the longest public chain and the longest fork generated by the adversary. Assume block $B$ is the forking point. Let $N_h(T)$ represent the number honest blocks proposed during $T$ amount of time after block $B$ added to the chain. Let $D_h(T)$ represent the number of blocks on the longest public chain starting from block $B$. And let $D_a(T)$ represents the number of blocks on the longest adversarial fork after $B$. There exists a constant $\zeta > 0$, such that:

\begin{equation} \label{eq:N_h_D_a_bound}
Pr(N_h(T) - D_a(T) > \frac{1}{3} (\lambda_h - e \lambda_a) T) \geq (1 - 2 e^{-\zeta T})^2
\end{equation}

\begin{equation} \label{eq:D_h_D_a_bound}
Pr(D_h(T) - D_a(T) > \frac{1}{3} (\lambda_h - e \lambda_a) T) \geq (1 - 2 e^{-\zeta T})^2
\end{equation}

\end{theorem}

Next, we prove that the three basic properties, i.e. chain growth, chain quality, and common prefix are guaranteed with high probability. For the proofs, we define a \textbf{round} as a fixed time interval $\Delta > 0$.

\bigbreak

\noindent \textbf{Chain Growth}. The chain growth property says that if an honest rational validator currently has a chain $C$, then after sufficiently large $s \in \mathbb{N}$ consecutive rounds, the validator will \textbf{adopt} a chain that is at least $\tau \cdot s$ blocks longer than $C$, where $\tau > 0$ is a chain growth parameter. Here ``adopting a chain $C$'' means that chain $C$ in the longest public chain in the local view of the validator.

\begin{lemma} \label{lm:chain-growth} The chain growth property holds for $\tau = \lambda_h \Delta (1 - \delta)$.
\end{lemma}

\begin{proof} According to Theorem \ref{th:honest-adversarial-gap-bound}, if the honest rational parties controls more than $1 - \frac{1}{1+e}$ speed-weighted-stake, no adversarial fork can have a higher expected rate of growth than the longest public chain. Moreover, based on the zero-latency assumption, a block proposed by an honest rational node will be delivered to all honest nodes instantly. Thus, once an honest validator proposes a new block, the longest public chain grows by at least one block. Hence the longest public chain will grow indefinitely. Since a honest rational validator always adopts the longest chain, in must adopt a chain at least as long as the longest public chain. On the other hand, according to Lemma \ref{lm:longest-public-chain-growth-rate-bound}, $Pr(D_h(s \Delta) > \lambda_h \Delta (1 - \delta) \cdot s) \geq 1 - 2 e^{- \lambda_h s \Delta \delta^2 / 3} $. Note that the right hand side approaches 1 exponentially as $s$ approaches infinity, which means the chain growth property holds for high probability.
\end{proof}

\bigbreak

\noindent \textbf{Chain Quality}. This property says that if an honest validator has adopted chain $C$, then in any sufficiently large $l \in \mathbb{N}$ consecutive blocks of $C$, the ratio of honest blocks is at least $\mu$, where $\mu \in (0, 1]$ is the chain quality parameter.

\begin{lemma} \label{lm:chain-quality} The chain quality property holds with $\mu > 0$.
\end{lemma}

\begin{proof} Again, based on the zero-latency assumption, an honest rational validator always adopts the longest public chain. Let us analyze the chain quality of the longest public chain. Assume there is no $\mu > 0$ such that the chain quality property holds true, then for any $l \in \mathbb{N}$, there exists a point in time when the last $l$ consecutive blocks of the longest public chain contain no honest block. Let us prove this is highly unlikely. Denote $\nu = \frac{1}{3} (\lambda_h - e \lambda_a)$. According to Lemma \ref{lm:block-rate-speed-weight-stake}, $\lambda_h > e \lambda_a$, and thus $\nu > 0$ (see also the proof for Theorem \ref{th:honest-adversarial-gap-bound}). Let $B_h$ represent the last honest block on the longest public chain, i.e. the $l$ adversarial blocks appends to the chain after $B_h$. Assume these $l$ consecutive adversarial blocks takes $T$ time to generate. For convenience, we also introduce $t_b = T / l$ to represent the average block time of these $l$ adversarial blocks. Denote the number of honest blocks proposed after $B_h$ by $N_h(t_b \cdot l)$. Also note that by definition, $D_a(T) = D_a(t_b \cdot l) = l$. Then, according to Theorem \ref{th:honest-adversarial-gap-bound}, $Pr(N_h(T) < l(1 + \nu t_b)) = Pr(N_h(t_b \cdot l) < l(1 + \nu t_b)) = Pr(N_h(t_b \cdot l) < D_a(t_b \cdot l) +  l \cdot \nu t_b) = Pr(N_h(t_b \cdot l) - D_a(t_b \cdot l) < + l \cdot \nu t_b) \leq 1 - (1 - 2 e^{-\zeta t_b \cdot l})^2$. Thus, as $l$ approaches infinity, $Pr(N_h(t_b \cdot l) < l(1 + \nu t_b))$ approaches zero. In other words, with high probability, the number of honest blocks produced is at least $l(1 + \nu t_b) > l$. Due to the zero-latency assumption, no two honest blocks share the same block height. Thus, there must be a public fork after block $B_h$ with at least $l(1 + \nu t_b)$ blocks. This contradicts with the statement that the \textit{longest public} chain contains $l$ consecutive adversarial blocks after $B_h$.
\end{proof}

\bigbreak

\noindent \textbf{Common Prefix}. Given a chain $C$, let us define $len(C)$ as the length of the chain. Also, denote the subchain formed by block $i$ to $j$ by $C[i,j]$. Consider two two chains $C_1$ and $C_2$. The \textit{common prefix} property says if $C_1$ and $C_2$ are adopted by two honest rational validators at round $r$ respectively, then there exists a $k \in \mathbb{N}$ such that $C_1[0, len(C_1)-k]$ must be a prefix of $C_2$. Likewise, $C_2[0, len(C_2)-k]$ is a prefix of $C_1$.

\begin{lemma} \label{lm:common-prefix} The common prefix property holds with sufficiently large $k$.
\end{lemma}

\begin{proof}
First, we note that under the zero network latency assumption, two honest blocks will not be proposed for the same height. This is because once an honest block is proposed, all the other honest rational validators will receive the block immediately. Then, the honest rational validators will stop solving the VDP for the current height and move to the next height. Thus, all the honest blocks must be generated for different heights.

Now we can prove the common prefix property by contradiction. If the common fix property does not hold, for the last $k$ block heights, $C_1$ and $C_2$ must be completely disjoint. Assume out of these $k$ block heights, there are $y$ honest blocks were generated. As discussed above, there $y$ honest blocks are generated for different heights. To keep $C_1$ and $C_2$ disjoint for $k$ block heights, the number of adversarial blocks $z$ must be at least as large as $y$, such that the adversarial blocks can completely override all the honest blocks on either $C_1$ or $C_2$. However, this contradicts with Theorem \ref{th:honest-adversarial-gap-bound}, which implies  $P(z \geq y)$ decreases exponentially with $k$. Thus, with high probability, the common prefix property holds for a sufficiently large $k$.
\end{proof}

\bigbreak

\noindent \textbf{Persistence and Liveness.} Persistence means that once a block is confirmed for all honest rational validators, the probability that it can be reverted is negligible as the chain grows. Liveness means that a valid transaction submitted to an honest validator will eventually be included in a block confirmed on all honest rational validators. The main result of Garay et al. is that the above three properties, i.e. chain growth, chain quality, and common prefix implies persistence and liveness \cite{garay2015bitcoinbackbone,bagaria2019poslc}. Thus, we have the following theorem. The proof is omitted here since it essentially follows the arguments of Garay et al. \cite{garay2015bitcoinbackbone}.

\begin{theorem} \label{th:persistence-and-liveness}
Our proposed protocol achieves both persistence and liveness.
\end{theorem}



\subsection{Energy Efficiency} 




In this section, we analyze the energy efficiency of the proposed protocol, and discuss why our protocol can reach consensus with much less energy consumption than proof-of-work based protocols. 

\begin{lemma} \label{lm:attacker-winning-probability} If the adversarial parties control less than $\frac{1}{1+e}$ fraction of the total speed-weighted-stake, then there exists a constant $\zeta > 0$, such that:

\begin{equation}
Pr(D_h(T) > D_a(T)) \geq (1 - 2 e^{-\zeta T})^2
\end{equation}

\end{lemma}

\begin{proof}
Denote $\nu = \frac{1}{3} (\lambda_h - e \lambda_a)$. Theorem \ref{th:honest-adversarial-gap-bound} states that $Pr(D_h(T) - D_a(T) > \nu) \geq (1 - 2 e^{-\zeta T})^2$. On the other hand, according to Lemma \ref{lm:block-rate-speed-weight-stake}, following the honest strategy, the block production rate of party is proportional to the total speed-weighted stake that party owns. Thus, we have $\lambda_h > e \lambda_a$ since the total speed-weighted stake of the honest parties is at least $e$ times of that of the adversary parties. Thus, $\nu = \frac{1}{3} (\lambda_h - e \lambda_a) > 0$.  As a result, we have the above inequality which characterizes the probability that the longest public chain is longer than the adversarial private chain. Apparently, this probability approaches 1 quickly as $T$ increases.
\end{proof}

\begin{theorem} \label{th:attacker-expected-return} Assume the adversarial parties fork the chain for $T$ amount of time, and denote their total block reward by $R_a(T)$. If the adversarial parties own less than $\frac{1}{1+e}$ fraction of the total speed-weighted-stake, their expected total rewards approaches zero asymptotically, i.e., $\lim_{T \rightarrow \infty}{\mathbb{E}[R_a(T)]} = 0$.
\end{theorem}

\begin{proof} To prove the claim, we first note that for the adversarial parties to win any reward, their private chain needs to be longer than the longest public chain. Moreover, the public chain and adversarial private chain grow independently. Thus, we have:

\begin{equation*}
\mathbb{E}[R_a(T)] = R\sum_{k=1}^{\infty}{Pr(D_a(T) = k)Pr(D_h(T) < k) \cdot k}
\end{equation*}

\noindent where $R$ is the reward for one block. To calculate the above summation, we separate it into two parts, the first with $k$ summing from 1 to $\floor{\lambda_h T}$, and the second part with $k$ summing from $\ceil{\lambda_h T}$ to infinity. For the first part, since event $\{D_a(T) = k, D_h(T) < k\}$ is a subset of $\{D_h(T) \leq D_a(T)\}$, we have $Pr(D_a(T) = k)Pr(D_h(T) < k) \leq 1 - Pr(D_h(T) > D_a(T)) \leq 1 - (1 - 2 e^{-\zeta T})^2$, where the last step is due to Lemma \ref{lm:attacker-winning-probability}. Thus, we have the following inequality:

\begin{align*}
&\sum_{k=1}^{\floor{\lambda_h T}}{Pr(D_a(T) = k)Pr(D_h(T) < k) \cdot k} \\
&\leq \sum_{k=1}^{\floor{\lambda_h T}}{(1 - (1 - 2 e^{-\zeta T})^2) \cdot k} = 2(e^{-\zeta T} - e^{-2\zeta T})(1 + \floor{\lambda_h T})\floor{\lambda_h T} \\
\end{align*}

It is straightforward to prove that the above value approaches zero as $T$ approaches infinity. Next, we look at the second part where $k \geq \ceil{\lambda_h T}$. For this, we note that obviously $Pr(D_h(T) < k) \leq 1$. Moreover, $Pr(D_a(T) = k) \leq Pr(D_a(T) > k - 1) < e^{-(k - 1 - e \lambda_a T)}$, where the last step is due to Lemma \ref{lm:bfs-speed-up-bound}. Thus, we can bound the second part of the summation by:

\begin{align*}
&\sum_{k=\ceil{\lambda_h T}}^{\infty}{Pr(D_a(T) = k)Pr(D_h(T) < k) \cdot k} \\
&\leq \sum_{k=\ceil{\lambda_h T}}^{\infty}{e^{-(k - 1 - e \lambda_a T)} \cdot 1 \cdot k} = \sum_{z=0}^{\infty}{e^{-z - \ceil{\lambda_h T} + 1 + e \lambda_a T} \cdot (z + \ceil{\lambda_h T})} \\
&\leq \sum_{z=0}^{\infty}{e^{-z - (\lambda_h - e \lambda_a) T + 1} \cdot (z + \lambda_h T + 1)} = e^{- (\lambda_h - e \lambda_a) T + 1} \sum_{z=0}^{\infty}{e^{-z} \cdot (z + \lambda_h T + 1)}
\end{align*}

Note that in the above derivation, for ease of calculation, we substitute $k$ with $z + \ceil{\lambda_h T}$, which does not change the result. In the proof of Lemma \ref{lm:attacker-winning-probability} we has mentioned that $(\lambda_h - e \lambda_a) T > 0$, thus, it is straightforward to prove that the above value also approaches zero as $T$ approaches infinity. Summing the two parts together, we can conclude that $\lim_{T \rightarrow \infty}{\mathbb{E}[R_a(T)]} = 0$.
\end{proof}

The above theorem indicates that without owning more than $\frac{1}{1+e}$ fraction of the total speed-weighted-stake, investing in extra VDP solvers to grow private block trees brings little return. Hence, purchasing VDP solvers to secretly produce a private chain is \textbf{worst off} than following the protocol as prescribed. Thus, unlike proof-of-work based protocols such as Bitcoin, our protocol can disincentive users from trapping into the hashing power arms race, and is therefore much more \textbf{energy efficient} while attains a high degree of decentralization.

\subsection{Long-Range Attack Resistance} 


One category of common attacks to PoS blockchains is the ``long-range attack", where the adversary forks the blockchain starting from a distant past (e.g. the genesis block), with the hope that the private branch will take over the longest public chain at some point. Due to the sequential dependency of the VDPs (i.e. the VDP for block height $l+1$ depends on the VDP solution contained in the block for height $l$), our protocol significantly increases the difficulty of long-range attack, since minting a chain requires solving VDPs sequentially. In particular, we have the following theorem (proof provided in Appendix \ref{sec:proofs-long-range-attack}):

\begin{theorem} \label{th:basic-long-range-resistance}
If an adversary is never able to acquire more than $\frac{1}{1+e}$ fraction of the total speed-weighted-stake at a certain point in history, then the probability that he can successfully launch a long-range attack is negligible. 
\end{theorem}

However, as discussed in the introduction, there is a more advanced form of long-range attack called the \textit{posterior corruption} attack. PoS protocols are constructed based on a assumption that stakeholders are motivated to keep the system running correctly since they have skins in the game. However, if a stakeholder has sold off his tokens, he no longer has this incentive. Once a sufficient portion of stakeholders from a point in time in the past are divested, there is a chance that they are willing to sell their private keys to an adversary. As the adversary acquires more past stakes, his private chain could grow faster. We note per Theorem \ref{th:basic-long-range-resistance}, if an adversary is not able to acquire more than $\frac{1}{1+e}$ fraction of total speed-weighted-stake at a certain block height in the past, his posterior corruption attack always fails. This is because his private fork can never catch up with the longest public chain, unlike in other PoS protocols without VDP. Below we discuss more general cases.

If an adversary can control more than $\frac{1}{1+e}$ fraction of total speed-weighted-stake, his private fork might indeed grow faster than the longest public chain. However, depending on the speed-weighted-stake he acquired, his private fork might not grow much faster than the longest public chain, and thus could take a long time before it catches up.  For quantitative analysis, let us model the validator staking events as a Poisson process where at each block height, any node has $p_{s}$ probability to stake tokens to become a validator. Then, we have the following result (proof provided in Appendix \ref{sec:proofs-long-range-attack}):

\begin{theorem} \label{th:adversary-stake-bound}
Assume the honest parties controls $\alpha_h$ fraction of stake at the latest confirmed block whose height is $l_c$. Then, at block height $l_a < l_c$, through posterior corruption, the expected maximum fraction of stake an adversary can potentially control is bounded by $p_{s} (l_c - l_a) \alpha_{h} + (1 - \alpha_{h})$.
\end{theorem}

For simplicity, let us analyze the scenario where the honest and adversary party have the same VDP solving speed. The more general cases can be analyzed similarly. Assume the adversary launches the posterior corruption based long-range attack from block height $l_a$. In the meanwhile, to increase the winning chance, he stops proposing blocks on the longest public chain. Thus, the longest chain only grows with average rate $\lambda_h$. On the other hand, the growth rate of the longest adversary fork is bounded by $e \cdot \frac{p_{s} (l_c - l_a) \alpha_{h} + (1 - \alpha_{h})}{\alpha_h} \cdot \lambda_h$. Let $T_a$ represent the minimum amount of time needed for the adversary fork to outgrow the longest public chain, the following condition needs to hold true:

\begin{equation*}
l_a + e \cdot \frac{p_{s} (l_c - l_a) \alpha_{h} + (1 - \alpha_{h})}{\alpha_h} \cdot \lambda_h T_a > l_c + \lambda_h T_a     
\end{equation*}

\noindent From this inequality, it is straightforward to derive that if the following inequality holds true, then the posterior corruption long-range attack from block height $l_a$ would have negligible chance to succeed. 

\begin{equation} \label{eq:long-range-fail-condition}
\begin{aligned}
    \alpha_h &> (1 + 1/e - p_s (l_c - l_a))^{-1}
\end{aligned}
\end{equation}

\noindent Otherwise, the attack would take at least 

\begin{equation} \label{eq:long-range-execution-time}
T_a > \frac{l_c - l_a}{\lambda_h (e p_s (l_c - l_a) + e / \alpha_h - e - 1)}
\end{equation}

\noindent amount of time to execute. To get a more concrete sense, let us plug in some numbers. Assume $\lambda_h = 144$ blocks/day (i.e. 10 minutes average block time similar to Bitcoin), and the honest stake fraction $\alpha_h = 0.9$, and $p_s = 1.903 \times 10^{-5}$ (i.e., on average an honest node stakes tokens for one year before unstaking, which is a reasonable estimation based on data points from a couple public PoS blockchains). Given these parameters, Inequality (\ref{eq:long-range-fail-condition}) holds true if $l_c - l_a < 13495$, which means the posterior corruption attack forking from any block less than 13495 blocks ago (approximately 93.7 days) would fail. Moreover, according to Inequality (\ref{eq:long-range-execution-time}), an attack forking from $l_c - l_a > 13495$ blocks ago would take at least $T_a > \frac{l_c - l_a}{0.00745 (l_c - l_a) - 100.507} = \frac{1}{0.00745 - 100.507 / (l_c - l_a)} \geq 1 / 0.00745 \approx 134.2$ days. For example, if the adversary forks from 50,000 blocks ago, the fork will take approximately 183.8 days to catch up the the longest public chain, which is a relatively long time. In comparison, for a PoS protocol without VDP, generating 50,000 consecutive blocks could just take minutes, if not seconds. Base on the analysis above, we can thus conclude that VDP effectively improves the long-range attack resistance for the consensus protocol.


\section{Conclusions and Future Works}

In this paper we introduced the concept of verifiable delay puzzle and proposed a new proof-of-stake/proof-of-delay hybrid consensus protocol built on top of it. It has the advantage of low energy consumption and can scale to accommodate tens of thousands of validators nodes similar to Bitcoin. It is also more resilient to the long-range attacks compared to existing PoS protocol proposals. However, we also recognize that there are large room for improvement. In particular, Security analysis of the protocol under more types of attacks such as the selfish-mining attack \cite{eyal2013selfishmining,gervais2016powsecurity} also needs to be conducted. Another extension is to incorporate other types fork choice rules such as the GHOST rule \cite{sompolinsky2015ghost}. Finally, transaction throughput improvement via alternative data structures like subchains \cite{rizun2016subchains} and DAG \cite{bagaria2018prism,yang2019prism10000,li2018conflux} would be an interesting area to explore.

%

%
%

\bibliographystyle{unsrt}
\bibliography{ms}

\begin{thebibliography}{10}

\bibitem{nakamoto2008bitcoin}
Satoshi Nakamoto.
\newblock Bitcoin: A peer-to-peer electronic cash system.
\newblock 2008.

\bibitem{bentov2014wopow}
Iddo Bentov, Ariel Gabizon, and Alex Mizrahi.
\newblock Cryptocurrencies without proof of work.
\newblock {\em CoRR}, abs/1406.5694, 2014.

\bibitem{daian2016snowwhite}
Phil Daian, Rafael Pass, and Elaine Shi.
\newblock Snow white: Provably secure proofs of stake.
\newblock Cryptology ePrint Archive, Report 2016/919, 2016.
\newblock \url{https://eprint.iacr.org/2016/919}.

\bibitem{buchman2018tendermint}
Ethan Buchman, Jae Kwon, and Zarko Milosevic.
\newblock The latest gossip on {BFT} consensus.
\newblock {\em CoRR}, abs/1807.04938, 2018.

\bibitem{buterin2016posfaq}
Vitalik Buterin.
\newblock Ethereum proof of stake faq.
\newblock \url{https://github.com/ethereum/wiki/wiki/Proofof-StakeFAQ}, 2016.

\bibitem{chen2018algorand}
Jing Chen, Sergey Gorbunov, Silvio Micali, and Georgios Vlachos.
\newblock Algorand agreement: Super fast and partition resilient byzantine
  agreement.
\newblock Cryptology ePrint Archive, Report 2018/377, 2018.
\newblock \url{https://eprint.iacr.org/2018/377}.

\bibitem{gilad2017algorand}
Yossi Gilad, Rotem Hemo, Silvio Micali, Georgios Vlachos, and Nickolai
  Zeldovich.
\newblock Algorand: Scaling byzantine agreements for cryptocurrencies.
\newblock In {\em Proceedings of the 26th Symposium on Operating Systems
  Principles}, pages 51--68. ACM, 2017.

\bibitem{micali2016algorand}
Silvio Micali.
\newblock {ALGORAND:} the efficient and democratic ledger.
\newblock {\em CoRR}, abs/1607.01341, 2016.

\bibitem{buterin2017casperffg}
Vitalik Buterin and Virgil Griffith.
\newblock Casper the friendly finality gadget.
\newblock {\em CoRR}, abs/1710.09437, 2017.

\bibitem{kiayias2017ouroboros}
Aggelos Kiayias, Alexander Russell, Bernardo David, and Roman Oliynykov".
\newblock Ouroboros: A provably secure proof-of-stake blockchain protocol.
\newblock In Jonathan Katz and Hovav Shacham", editors, {\em Advances in
  Cryptology -- CRYPTO 2017}, pages 357--388, Cham, 2017. Springer
  International Publishing.

\bibitem{david2017ouroborospraos}
Bernardo David, Peter Ga{\v{z}}i, Aggelos Kiayias, and Alexander Russell.
\newblock Ouroboros praos: An adaptively-secure, semi-synchronous
  proof-of-stake protocol.
\newblock Cryptology ePrint Archive, Report 2017/573, 2017.
\newblock \url{https://eprint.iacr.org/2017/573}.

\bibitem{badertscher2018ouroborosgenesis}
Christian Badertscher, Peter Gazi, Aggelos Kiayias, Alexander Russell, and
  Vassilis Zikas.
\newblock Ouroboros genesis: Composable proof-of-stake blockchains with dynamic
  availability.
\newblock Cryptology ePrint Archive, Report 2018/378, 2018.
\newblock \url{https://eprint.iacr.org/2018/378}.

\bibitem{lamport1982bgp}
Leslie Lamport, Robert Shostak, and Marshall Pease.
\newblock The byzantine generals problem.
\newblock {\em ACM Trans. Program. Lang. Syst.}, 4(3):382--401, July 1982.

\bibitem{fischer1985flp}
Michael~J. Fischer, Nancy~A. Lynch, and Michael~S. Paterson.
\newblock Impossibility of distributed consensus with one faulty process.
\newblock {\em J. ACM}, 32(2):374--382, April 1985.

\bibitem{dwork1988consensuspartialsync}
Cynthia Dwork, Nancy Lynch, and Larry Stockmeyer.
\newblock Consensus in the presence of partial synchrony.
\newblock {\em J. ACM}, 35(2):288--323, April 1988.

\bibitem{castro1999pbft}
Miguel Castro and Barbara Liskov.
\newblock Practical byzantine fault tolerance.
\newblock In {\em Proceedings of the Third Symposium on Operating Systems
  Design and Implementation}, OSDI '99, pages 173--186, Berkeley, CA, USA,
  1999. USENIX Association.

\bibitem{abraham2018hotstuff}
Ittai Abraham, Guy Gueta, and Dahlia Malkhi.
\newblock Hot-stuff the linear, optimal-resilience, one-message {BFT} devil.
\newblock {\em CoRR}, abs/1803.05069, 2018.

\bibitem{deirmentzoglou2019longrangesurvey}
Evangelos Deirmentzoglou, Georgios Papakyriakopoulos, and Constantinos
  Patsakis.
\newblock A survey on long-range attacks for proof of stake protocols.
\newblock {\em IEEE Access}, PP, 02 2019.

\bibitem{franklin2006keyevolving}
Matt Franklin.
\newblock A survey of key evolving cryptosystems.
\newblock {\em Int. J. Security and Networks}, 1, 2006.

\bibitem{boneh2018vdf}
Dan Boneh, Joseph Bonneau, Benedikt B{\"u}nz, and Ben Fisch.
\newblock Verifiable delay functions.
\newblock In {\em Annual International Cryptology Conference}, pages 757--788.
  Springer, 2018.

\bibitem{pietrzak2018svdf}
Krzysztof Pietrzak.
\newblock Simple verifiable delay functions.
\newblock Cryptology ePrint Archive, Report 2018/627, 2018.
\newblock \url{https://eprint.iacr.org/2018/627}.

\bibitem{wesolowski2018evdf}
Benjamin Wesolowski.
\newblock Efficient verifiable delay functions.
\newblock Cryptology ePrint Archive, Report 2018/623, 2018.
\newblock \url{https://eprint.iacr.org/2018/623}.

\bibitem{boneh2018survey2vdfs}
Dan Boneh, Benedikt B\"unz, and Ben Fisch.
\newblock A survey of two verifiable delay functions.
\newblock Cryptology ePrint Archive, Report 2018/712, 2018.
\newblock \url{https://eprint.iacr.org/2018/712}.

\bibitem{ephraim2020cvdf}
Naomi Ephraim, Cody Freitag, Ilan Komargodski, and Rafael Pass.
\newblock Continuous verifiable delay functions.
\newblock In Anne Canteaut and Yuval Ishai, editors, {\em Advances in
  Cryptology -- EUROCRYPT 2020}, pages 125--154, Cham, 2020. Springer
  International Publishing.

\bibitem{cohen2016post}
Bram Cohen.
\newblock Proofs of space and time - removing waste.
\newblock Blockchain Protocol Analysis and Security Engineering, 2017.
\newblock \url{https://cyber.stanford.edu/sites/default/files/bramcohen.pdf}.

\bibitem{cohen2019chia}
Bram Cohen and Krzysztof Pietrzak.
\newblock The chia network blockchain.
\newblock \url{https://www.chia.net/assets/ChiaGreenPaper.pdf}, 2019.

\bibitem{drake2018randaomvp}
Justin Drake.
\newblock {Minimal VDF randomness beacon}.
\newblock \url{https://ethresear.ch/t/minimal-vdf-randomness-beacon/3566},
  2018.

\bibitem{drake2018vdflookahead}
Justin Drake.
\newblock {VDF-based RNG with linear lookahead}.
\newblock
  \url{https://ethresear.ch/t/vdf-based-rng-with-linear-lookahead/2573}, 2018.

\bibitem{park2018spacemint}
Sunoo Park, Krzysztof Pietrzak, Albert Kwon, Jo{\"e}l Alwen, Georg Fuchsbauer,
  and Peter Gazi.
\newblock Spacemint: A cryptocurrency based on proofs of space.
\newblock {\em Financial Cryptography and Data Security}, 2018.

\bibitem{garay2015bitcoinbackbone}
Juan Garay, Aggelos Kiayias, and Nikos Leonardos.
\newblock The bitcoin backbone protocol: Analysis and applications.
\newblock In Elisabeth Oswald and Marc Fischlin, editors, {\em Advances in
  Cryptology - EUROCRYPT 2015}, pages 281--310, Berlin, Heidelberg, 2015.
  Springer Berlin Heidelberg.

\bibitem{bagaria2019poslc}
Vivek Bagaria, Amir Dembo, Sreeram Kannan, Sewoong Oh, David Tse, Pramod
  Viswanath, Xuechao Wang, and Ofer Zeitouni.
\newblock Proof-of-stake longest chain protocols: Security vs predictability,
  2019.

\bibitem{eyal2013selfishmining}
Ittay Eyal and Emin~G{\"{u}}n Sirer.
\newblock Majority is not enough: Bitcoin mining is vulnerable.
\newblock {\em CoRR}, abs/1311.0243, 2013.

\bibitem{gervais2016powsecurity}
Arthur Gervais, Ghassan~O Karame, Karl W{\"u}st, Vasileios Glykantzis, Hubert
  Ritzdorf, and Srdjan Capkun.
\newblock On the security and performance of proof of work blockchains.
\newblock In {\em Proceedings of the 2016 ACM SIGSAC conference on computer and
  communications security}, pages 3--16. ACM, 2016.

\bibitem{sompolinsky2015ghost}
Yonatan Sompolinsky and Aviv Zohar.
\newblock Secure high-rate transaction processing in bitcoin.
\newblock In {\em International Conference on Financial Cryptography and Data
  Security}, pages 507--527. Springer, 2015.

\bibitem{rizun2016subchains}
Peter~R Rizun.
\newblock Subchains: A technique to scale bitcoin and improve the user
  experience.
\newblock {\em Ledger}, 1:38--52, 2016.

\bibitem{bagaria2018prism}
Vivek Bagaria, Sreeram Kannan, David Tse, Giulia Fanti, and Pramod Viswanath.
\newblock Deconstructing the blockchain to approach physical limits.
\newblock Cryptology ePrint Archive, Report 2018/992, 2018.
\newblock \url{https://eprint.iacr.org/2018/992}.

\bibitem{yang2019prism10000}
Lei Yang, Vivek Bagaria, Gerui Wang, Mohammad Alizadeh, David Tse, Giulia
  Fanti, and Pramod Viswanath.
\newblock Prism: Scaling bitcoin by 10, 000x.
\newblock {\em ArXiv}, abs/1909.11261, 2019.

\bibitem{li2018conflux}
Chenxing Li, Peilun Li, Wei Xu, Fan Long, and Andrew~Chi{-}Chih Yao.
\newblock Scaling nakamoto consensus to thousands of transactions per second.
\newblock {\em CoRR}, abs/1805.03870, 2018.

\end{thebibliography}

%
%

\begin{subappendices}
\renewcommand{\thesection}{\Alph{section}}%

\section{Analysis of the Slashing Rules} \label{sec:analysis-of-slashing-rules}

In this Appendix, we will prove that as stated by Theorem \ref{th:nothing-at-stake-immune} in Section \ref{sec:slashing-rules}, for the rational nodes, when there are forks, extending all the forks is never the best strategy. Instead, the dominant strategy is to only extend the longest fork in its local view. This result indicates that our protocol is robust against the ``nothing-at-stake" attacks.

\bigbreak
\noindent \textbf{A Simple Example}
\bigbreak

\begin{figure}[h!]
\centering
\includegraphics[width=0.85\textwidth]{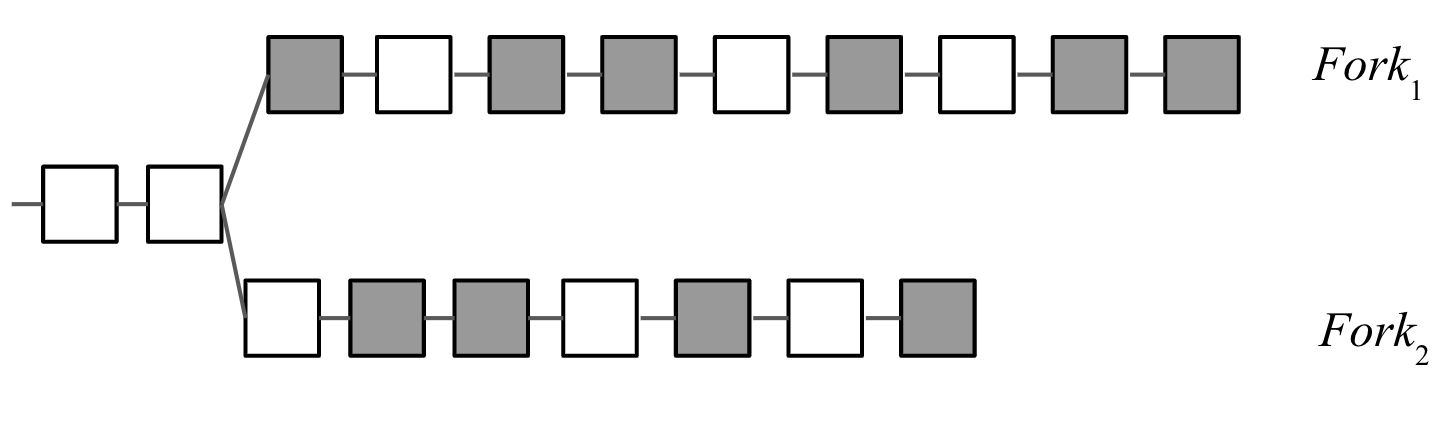}
\caption{Simple example with two forks. Validator $v$ can choose to publish 6 blocks (blocks in gray) on $Fork_1$, and 4 blocks on $Fork_2$. Or it can just stick to one single fork.}
\label{fig:slash_simple_case}
\end{figure}

First, let us use the simple case in Fig. \ref{fig:slash_simple_case} to demonstrate our slashing rules. In this example, the blockchain has two forks. The probability that $Fork_1$ wins is 0.7, and hence the probability that $Fork_2$ wins is 0.3. A validator $v$ can mine on both forks, where $Fork_1$ contains 6 blocks it mines, while $Fork_2$ contains 4 blocks. Alternatively, it can publish 6 blocks on $Fork_1$ and no block on $Fork_2$. It can also just mine 4 blocks on $Fork_2$ only.

If it mines on both forks, based on our proposed slashing rules, the expected reward is $0.7 \cdot (6R - 4(1 + \epsilon)R) + 0.3 \cdot (4R - 6(1 + \epsilon)R) = (0.8 - 4.6 \epsilon)R$. To maximize the total reward, validator $v$ should try to submit the slashing transactions by itself on both forks to claim the submitter reward. This is a bit counter-intuitive, but if $v$ does not submit the slashing transactions, someone else will do so, and $v$ is worst off. Hence, its expected maximum total reward by mining on both forks is $0.8R$. Yet if he only mines on $Fork_1$, his expected maximum total reward will be $0.7 \cdot 6R + 0.3 \cdot (-6R) = 2.4R$. Clearly, in this example, validator $v$ is better off if he only publishes blocks on $Fork_1$.

\bigbreak
\noindent \textbf{The General Two Forks Case}
\bigbreak

\begin{figure}[h!]
\centering
\includegraphics[width=0.85\textwidth]{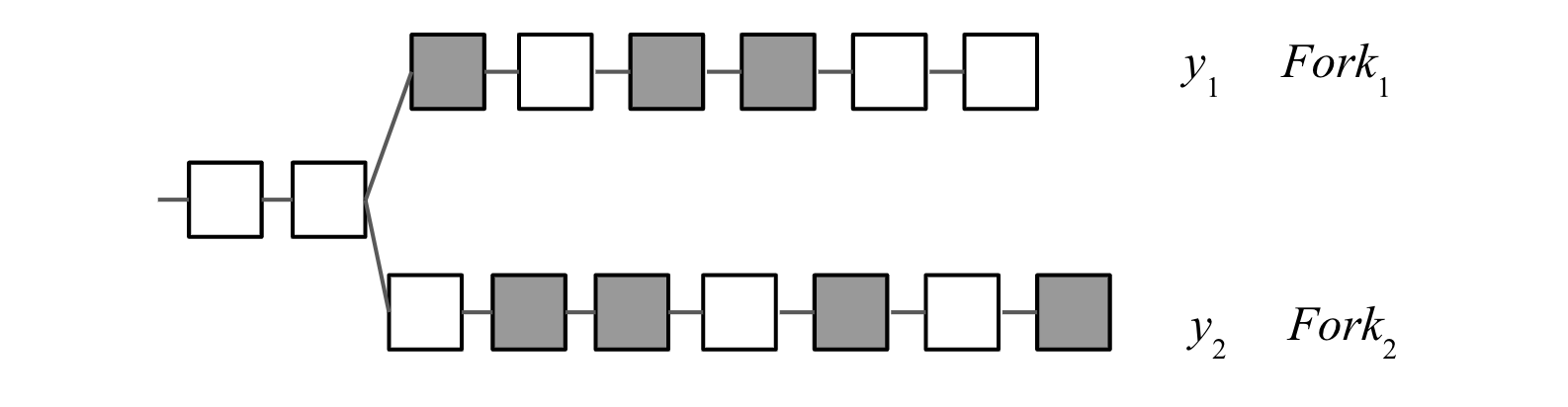}
\caption{The general two forks scenario. Similarly, validator $v$ has the options to publish on both forks, or stick to just one fork.}
\label{fig:slash_general_2_forks}
\end{figure}

Fig. \ref{fig:slash_general_2_forks} illustrates a scenario similar to Fig. \ref{fig:slash_simple_case}, but does not give the concrete winning probably of each fork. When validator $v$ sees the two forks, it can have three strategies:

\begin{itemize}
\item \textbf{Strategy 0}. Publish blocks on both forks, assume for the next $N$ block heights, it can mine $y_1$ blocks on $Fork_1$, and $y_2$ blocks on $Fork_2$.
\item \textbf{Strategy 1}. Publish $y_1$ blocks on $Fork_1$ only.
\item \textbf{Strategy 2}. Publish $y_2$ blocks on $Fork_2$ only.
\end{itemize}

Let us calculate the expected reward for each strategy. Suppose $p$ is the probability that $Fork_1$ becomes the winning chain. Then the probability that $Fork_2$ wins will be $(1-p)$. For \textbf{Strategy 0}, its total expected reward after the forking point is

\begin{equation} \label{eq:two_forks_strategy_0_reward}
p \cdot (y_1 - y_2 (1 + \epsilon)) \cdot R + (1 - p) \cdot (y_2 - y_1 (1 + \epsilon)) \cdot R 
\end{equation}

This is because if $Fork_1$ wins, $v$ gets $y_1 \cdot R$ block rewards, but its deposit gets slashed by $y_2 (1 + \epsilon) \cdot R$. Similarly if $Fork_2$ wins, $v$ gets $y_2 \cdot R$ block rewards, and yet its deposit gets deducted by $y_1 (1 + \epsilon) \cdot R$. Similar to the simple example, to maximize the total reward, validator $v$ should try to submit the slashing transactions by itself on both forks. Thus, for \textbf{Strategy 0}, after simplifying Formula (\ref{eq:two_forks_strategy_0_reward}), the maximum total expected reward is 

\begin{equation} \label{eq:two_forks_strategy_0_reward_simplified}
  E_{0}[R_{total}] = (2p - 1) \cdot (y_1 - y_2) \cdot R
\end{equation}

For \textbf{Strategy 1}, since $v$ only mines on $Fork_1$, if $Fork_1$ wins, it gets $y_1 \cdot R$ block rewards with no deposit slash. However, if $Fork_2$ wins, it obtains no block reward and yet gets slashed by $y_1 (1 + \epsilon) \cdot R$. If $v$ also submits the slashing transaction to claim the submitter reward, the maximum total expected reward will be $(p \cdot y_1 - (1 - p) \cdot y_1) \cdot R$, which simplifies to

\begin{equation} \label{eq:two_forks_strategy_1_reward}
  E_{1}[R_{total}] = (2p - 1) \cdot y_1 \cdot R
\end{equation}

Similarly, the maximum total expected reward for \textbf{Strategy 2} can be calculated by

\begin{equation} \label{eq:two_forks_strategy_2_reward}
  E_{2}[R_{total}] = (1 - 2p) \cdot y_2 \cdot R
\end{equation}

Notice that according to Formula (\ref{eq:two_forks_strategy_0_reward_simplified}-\ref{eq:two_forks_strategy_2_reward}), $E_{0}[R_{total}] = E_{1}[R_{total}] + E_{2}[R_{total}]$. On the other hand, $E_{1}[R_{total}] \cdot E_{2}[R_{total}] = - (2p - 1)^2 \cdot y_1 \cdot y_2 \cdot R^2 \leq 0$. This means that $E_{1}[R_{total}]$ and $E_{2}[R_{total}]$ can not be positive at the same time. Hence

\begin{equation}
  E_{0}[R_{total}] \leq \max \left\{E_{1}[R_{total}], E_{2}[R_{total}]\right\}
\end{equation}

Thus, mining on both forks is not the best strategy for validator $v$. This results can be extended to the cases with \textbf{more than two forks} easily with similar reasoning.

\bigbreak
\noindent \textbf{Multiple Disjoint Forks with Blocks Already Published}
\bigbreak

\begin{figure}[h!]
\centering
\includegraphics[width=0.85\textwidth]{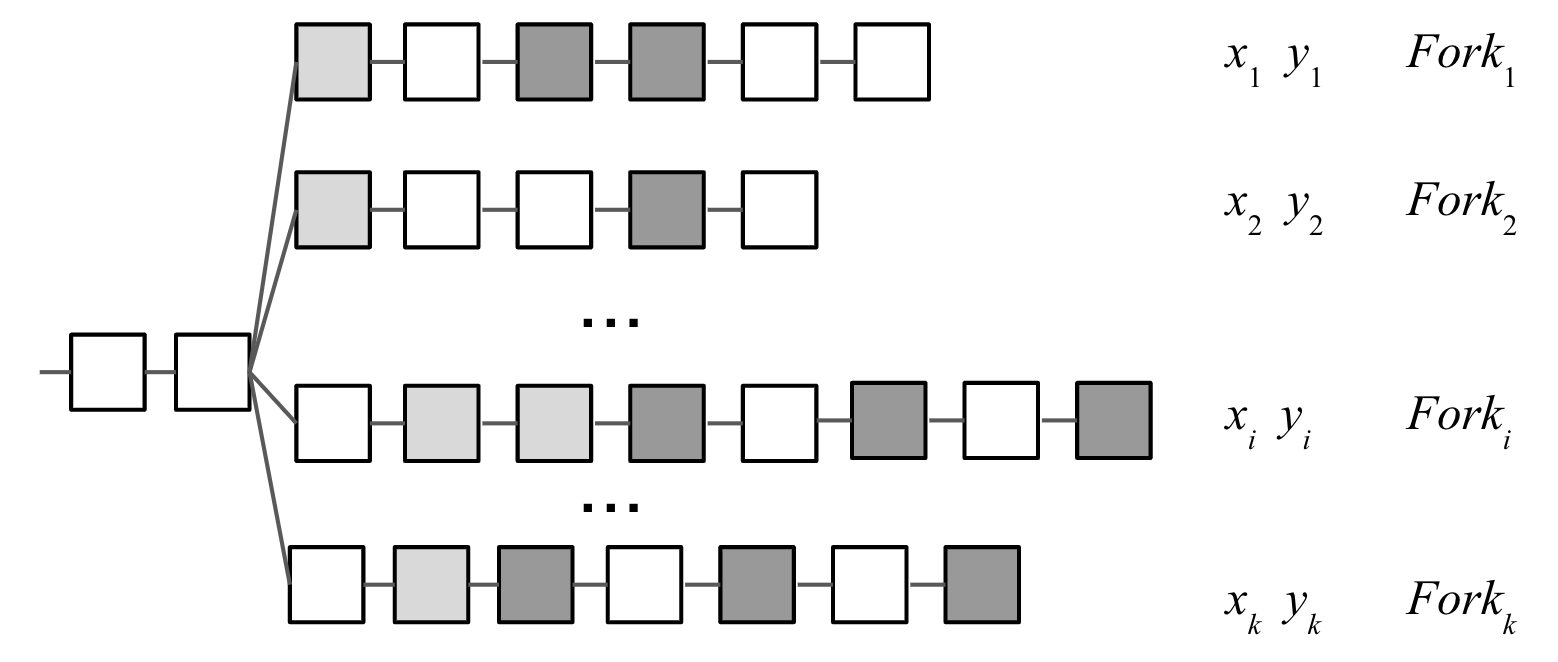}
\caption{The general multiple disjoint forks case. The light gray blocks represent the blocks validator $v$ has already published. And the dark gray blocks represent the blocks $v$ can choose to publish.}
\label{fig:slash_general_multiple_disjoint_forks}
\end{figure}

Fig. \ref{fig:slash_general_multiple_disjoint_forks} depicts a more general case where there are $k$ disjoint forks, and validator $v$ has already mined $x_i$ blocks on $Fork_i$. Now it has the following $k+1$ strategies: 

\begin{itemize}
\item \textbf{Strategy 0}. Publish blocks on all forks, assume for the next $N$ block heights, it can mine $y_i$ more blocks on $Fork_i$, in addtion to the $x_i$ blocks it already published on $Fork_i$.
\item \textbf{Strategy \textit{i}}. Publish $y_i$ more blocks on $Fork_i$ only.
\end{itemize}

Let us first analyze \textbf{Strategy 0}. Assuming the winning probability of $Fork_i$ is $p_i$, the maximum total expected reward can be calculated by

\begin{equation} \label{eq:multi_disjoint_forks_strategy_0_reward}
  E_{0}[R_{total}] = \sum^{k}_{i=1}{p_i \cdot \left(x_i + y_i - (X - x_i + Y - y_i)\right) \cdot R}
\end{equation}

Here $X = \sum^{k}_{i=1}{x_i}$, and $Y = \sum^{k}_{i=1}{y_i}$. Formula (\ref{eq:multi_disjoint_forks_strategy_0_reward}) can be simplified as

\begin{equation}
  E_{0}[R_{total}] = \left(2\sum^{k}_{i=1}{p_i \cdot (x_i + y_i)} - (X + Y)\right) \cdot R
\end{equation}

On the other hand, the maximum total expected reward of \textbf{Strategy \textit{i}} can be calculated by

\begin{equation}
  E_{i}[R_{total}] = \left(p_i \cdot \left(x_i + y_i - (X - x_i)\right) + \sum_{j \neq i}{p_j \cdot \left(x_j - y_i - (X - x_j)\right)} \right) \cdot R
\end{equation}

Simplifying this formula gets us

\begin{equation}
  E_{i}[R_{total}] = \left(2 \sum^{k}_{j=1}{p_j \cdot (x_j + y_j)} - \left(X + 2 \sum_{j \neq i}{p_j \cdot y_j} + y_i \right) \right) \cdot R
\end{equation}

Thus we have

\begin{equation}
  E_{i}[R_{total}] - E_{0}[R_{total}] = \left(\sum_{j \neq i}{(1 - 2 p_j) \cdot y_j} \right) \cdot R
\end{equation}

If $p_j < 1/2$ for all $j = 1, 2, .., k$, obviously $\left(\sum_{j \neq i}{(1 - 2 p_j) \cdot y_j} \right) \cdot R \geq 0$ for any $i$. Hence, for $E_{i}[R_{total}] \geq E_{0}[R_{total}]$ for any $i$. Otherwise, if there is a $i$ such that $p_i \geq 1/2$, then for all $j \neq i$, we must have $p_j \leq 1/2$, since $\sum_{j=1}{p_j} = 1$. Thus, for this $i$, $\left(\sum_{j \neq i}{(1 - 2 p_j) \cdot y_j} \right) \cdot R \geq 0$. Therefore, in any case, there must be a \textbf{Strategy \textit{i}} whose maximum total expected reward is at least as large as the mining-on-all-forks strategy. This means when there are multiple disjoint forks, even if a validator has already published some blocks on multiple forks, the best strategy is still to \textit{mine only on the fork with the largest winning probability}.

\section{Proofs for the Protocol Backbone Analysis} \label{sec:proofs-backbone}

Below we prove \textbf{Lemma \ref{lm:block-rate-speed-weight-stake}}:

\begin{proof}
Recall that the honest behavior can be viewed as a greedy strategy. The validators compete to solve the assigned VDP for the current block height. If one validator solves its assigned VDP for the current height and publish a new block, all other validators give up on the current VDP and move to VDP for the next block height. Note that such a VDP solving based block production process can be modeled as a Poisson process, where the success rate of each trial of $t$ is $\gamma$ per Formula \ref{eq:hash_threshold}. Thus, similar to Bitcoin mining, equipping with a solver $\eta$ times as fast as others is equivalent to have $\eta$ validators working in parallel, which increase the block production rate by $\eta$ times. Hence, as a Poisson process, the aggregated block production rate $\lambda$ (i.e. the total number of blocks proposed by all validators in a unit time) is proportional to $W = \sum_{i \in V}{sws_i}$.
\end{proof}

Below we prove \textbf{Lemma \ref{lm:number-of-honest-blocks-bound}}:

\begin{proof}
First we note that with the zero-latency assumption, the VDP solving process can be modeled as independent Bernoulli trials. Thus, the adversarial block production does not affect the honest block production. Thus, the honest block production process can be modeled as a Poisson process with rate $\lambda_h$. Thus, by using the Chernoff bound, it is straightforward to derive that $Pr(N_h(T) < (1 - \delta) \lambda_h T) \leq 2 e^{- \lambda_h T \delta^2 / 3}$.
\end{proof}

Below we prove \textbf{Lemma \ref{lm:longest-public-chain-growth-rate-bound}}:

\begin{proof}
Assume for a moment that all the adversarial parties stop producing blocks after block $B$. In this case, the honest block production process can be modeled as a Poisson process with rate $\lambda_h$. And the proposed honest blocks will form a chain. Using the Chernoff bound, it is straightforward to derive that for this hypothetical honest chain, $Pr(D^{\prime}_h(T) < (1 - \delta) \lambda_h T) \leq 2 e^{- \lambda_h T \delta^2 / 3}$, where $D^{\prime}_h(T)$ is the number of block after block $B$ on this hypothetical chain after time period $T$.

On the other hand, due to the zero-network latency assumption, even in the presence of adversarial validators, two honest blocks will never share the same block height. This is because once a honest validator publish a new block, all other honest validators will receive it instantly. Then, all these honest validators move to solve the VDP for the next block height. Hence, the longest chain should grow at least as fast as the hypothetical chain. Thus, $Pr(D_h(T) < (1 - \delta) \lambda_h T) \leq Pr(D^{\prime}_h(T) < (1 - \delta) \lambda_h T) \leq 2 e^{- \lambda_h T \delta^2 / 3}$.
\end{proof}

Below we prove \textbf{Theorem \ref{th:honest-adversarial-gap-bound}}:

\begin{proof}

First, since the total honest speed-weighted-stake $W_h > 1 - \frac{1}{1+e}$, we have $W_h / W_a > (1 - \frac{1}{1+e}) / (1 - (1 - \frac{1}{1+e})) = e$, meaning the speed-weighted-stake controlled by honest rational parties is at least $e$ time as much as that controlled by the adversarial parties. Therefore, according to Lemma \ref{lm:block-rate-speed-weight-stake}, we have $\lambda_h > e \lambda_a$. Further, let us denote constant $\nu = \frac{1}{3} (\lambda_h - e \lambda_a)$. Then, we have the following

\begin{align*} 
  &     Pr(N_h(T) - D_a(T) > \nu T) \\
  &\geq Pr(\{N_h(T) > \lambda_h T - \nu T\} \land \{D_a(T) < e \lambda_a T + \nu T\}) \\
  & =   Pr(N_h(T) > \lambda_h T - \nu T) \cdot Pr(D_a(T) < e \lambda_a T + \nu T) \\
  & =   (1 - Pr(N_h(T) < \lambda_h T - \nu T)) \cdot (1 - Pr(D_a(T) > e \lambda_a T + \nu T)) \\
  &\geq (1 - 2 e^{-\frac{1}{3 \lambda_h} \nu^2 T}) \cdot (1 - e ^ {- \nu T}) \\
  &\geq (1 - 2 e^{-\zeta T})^2
\end{align*}

\noindent where $\zeta = \max \{\frac{1}{3 \lambda_h} \nu^2, \nu \}$ is a positive constant. In the above derivation, the first step is because the joint event $\{N_h(T) > \lambda_h T - \nu T\} \land \{D_a(T) < e \lambda_a T + \nu T\}$ is a sufficient condition for event $\{N_h(T) - D_a(T) > \nu T\}$. Moreover, $\{N_h(T) > \lambda_h T - \nu T\}$ and $\{D_a(T) < e \lambda_a T + \nu T\}$ are independent events since the adversarial fork after $B$ does not contain any honest block. This yields the second step. We then apply Inequality (\ref{eq:D_h_t_bound}) and (\ref{eq:D_a_t_bound}). The last step is obviously since $e^{-x}$ is a homogeneously decreasing function of $x$. This proves Inequality (\ref{eq:N_h_D_a_bound}).

For Inequality (\ref{eq:D_h_D_a_bound}), we just need to note that as discussed earlier, due to the zero-latency assumption, no two honest blocks share the same block height. Thus, during time $T$, the longest public chain must have grown by at least $N_h(T)$ blocks, i.e. $D_h(T) > N_h(T)$. Hence, we have $Pr(D_h(T) - D_a(T) > \nu T) \geq Pr(N_h(T) - D_a(T) > \nu T) \geq (1 - 2 e^{-\zeta T})^2$. This proves Inequality (\ref{eq:D_h_D_a_bound}).

Thus, if the adversary forks the chain from a certain block, with high probability, the gap between the longest public chain and the longest adversarial private fork increases with time.
\end{proof}





\section{Proofs for the Long-Range Attack Analysis} \label{sec:proofs-long-range-attack}

Below we prove \textbf{Thereom \ref{th:basic-long-range-resistance}}:

\begin{proof}
This theorem can be derived directly from Theorem \ref{th:honest-adversarial-gap-bound}, which states $Pr(D_h(T) - D_a(T) > \frac{1}{3} (\lambda_h - e \lambda_a) T) \geq (1 - 2 e^{-\zeta T})^2$. Assume an adversary tries to launch a long-rage attack from a block generated $T$ time ago, then the probability that his private fork is longer than the the longest public chain decreases exponentially with $T$.
\end{proof}

Below we prove \textbf{Thereom \ref{th:adversary-stake-bound}}:

\begin{proof}
We make the reasonable assumption that if a validator staked before height $l_a$ and remain staked at $l_c$, then the adversary cannot acquire its private key, since otherwise the adversary would have control to the tokens the validator currently owns. Let us calculate the expected fraction of stake that were staked before $l$ and remain staked at $l_c$ out of the $\alpha_h$ fraction of stake the honest parties currently controlled. For the calculation, we introduce random variable $n_{v,l}$, which is 1 if honest validator $v$ conducted the staking at block height $l$, and 0 otherwise. Then the total number of honest staking events between height $l_a$ and $l_c$ can be calculated by $N_s = \sum^{l_c-1}_{l=l_a} \sum_{v \in V_h} n_{v, l}$, where $V_h$ is the set of honest validators at height $l_c$. Thus, the expected number $E[N_s] = E[\sum^{l_c-1}_{l=l_a} \sum_{v \in V_h} n_{v, l}] = \sum^{l_c-1}_{l=l_a} \sum_{v \in V_h} E[n_{v, l}] =  \sum^{l_c-1}_{l=l_a} \sum_{v \in V_h} p_s = p_s (l_c - l_a) N_h$. Here $N_h = |V_h|$ is the number of honest validators at height $l_c$. On the other hand, since at height $l_c$ the honest validators controls $\alpha_h$ fraction of stake, the total fraction of stake that the adversary can potentially acquire is \textbf{at most} $\frac{E[N_s]}{N_h} \cdot \alpha_h = p_{s} (l_c - l_a) \alpha_{h}$ (``at most'' because if a node $v$ stakes and unstakes the same token twice during $l_a$ to $l_c$, the adversary can only acquire one token). This, plus the $1 - \alpha_{h}$ fraction of stake the adversary already owned without bribing at height $l_a$, gets us the $p_{s} (l_c - l_a) \alpha_{h} + (1 - \alpha_{h})$ bound. 
\end{proof}

\end{subappendices}

\end{document}